\documentclass[11pt]{article}

\usepackage[utf8]{inputenc}
\usepackage[T1]{fontenc}
\usepackage[english]{babel}

\usepackage{thmtools}
\usepackage{thm-restate}

\usepackage{amsmath,amssymb,amsfonts,amsthm}
\usepackage{mathabx}

\usepackage{cleveref}

\newtheorem{lemma}{Lemma}
\newtheorem{theorem}{Theorem}

\theoremstyle{definition}

\crefname{lemma}{Lemma}{Lemmas}
\Crefname{lemma}{Lemma}{Lemmas}
\crefname{theorem}{Theorem}{Theorems}
\Crefname{theorem}{Theorem}{Theorems}
\crefname{corollary}{Corollary}{Corollaries}
\Crefname{corollary}{Corollary}{Corollaries}
\crefname{observation}{Observation}{Observations}
\Crefname{observation}{Observation}{Observations}
\crefname{definition}{Definition}{Definitions}
\Crefname{definition}{Definition}{Definitions}
\crefname{section}{Section}{Sections}
\Crefname{section}{Section}{Sections}
\crefname{figure}{Figure}{Figures}
\Crefname{figure}{Figure}{Figures}

\renewcommand{\subset}{\subseteq}

\newcommand{\ceil}[1]{\left\lceil{#1}\right\rceil}
\newcommand{\floor}[1]{\left\lfloor{#1}\right\rfloor}

\newcommand{\abs}[1]{\left | #1 \right |}
\newcommand{\set}[1]{\left \{ #1 \right \}}

\newcommand{\norm}[1]{\abs{\abs{#1}}}
\newcommand{\Z}{\mathbb Z} 

\newcommand{\proj}[2]{\left [ #2 \right ]_{#1}}
\newcommand{\emb}[2]{\iota_{#1}\!\left( #2 \right)}
\newcommand{\embDef}[1]{\iota_{#1}}
\newcommand{\modNorm}[2]{\norm{#2}_{#1}}
\newcommand{\A}{\mathcal{A}}
\newcommand{\E}{\operatorname{E}}
\newcommand{\Prp}[1]{\Pr\!\left[{#1} \right]}
\newcommand{\Ep}[1]{\E\!\left[{#1} \right]}
\newcommand{\Oh}[1]{O\!\left(#1\right)}
\newcommand{\Omegah}[1]{\Omega\!\left(#1\right)}
\newcommand{\hb}{\bar{h}}

\usepackage{listings}
\usepackage{color}
\usepackage[pdftex]{graphicx}
\usepackage[usenames,dvipsnames,table]{xcolor}
\definecolor{shade}{RGB}{235,235,235}

\usepackage{booktabs}

\usepackage{todonotes}

\usepackage{authblk}

\usepackage{algorithm2e}

\title{Linear Hashing is Awesome}
\author{Mathias Bæk Tejs Knudsen\thanks{
	Research partly supported by
	Advanced Grant DFF-0602-02499B from the Danish Council for Independent Research
	under the Sapere Aude research career programme
	and by the FNU project AlgoDisc - Discrete Mathematics, Algorithms, and Data Structures}}
\affil{University of Copenhagen,\\
    \tt{mathias@tejs.dk}
}
\date{}

\begin{document}
\maketitle
\begin{abstract}
We consider the hash function $h(x) = ((ax+b) \bmod p) \bmod n$ where
$a,b$ are chosen uniformly at random from $\set{0,1,\ldots,p-1}$. We prove
that when we use $h(x)$ in hashing with chaining to insert $n$ elements
into a table of size $n$ the expected length of the longest chain is
$\tilde{O}\!\left(n^{1/3}\right)$. The proof also generalises to give the same bound
when we use the multiply-shift hash function by Dietzfelbinger et al. [Journal of Algorithms 1997].
\end{abstract}
\thispagestyle{empty}

\newpage
\setcounter{page}{1}

\section{Introduction}

In this paper we study the hash function $h : [p] \to [m]$ (where $[m] = \set{0,1,\ldots,m-1}$)
defined by $h(x) = ((ax+b) \bmod p) \bmod m$, where $a,b \in [p]$ are chosen uniformly at random
from $[p]$. Here, $p$ is a prime and $p \ge m$. We assume that we have a set $X \subset [p]$ of $n$
\emph{keys} with $n \le m$ and use $h$ to assign a hash value $h(x)$ to each key $x \in X$. We
are interested in the frequency of the most popular hash value, i.e. we study the random variable
$M(h,X)$ defined by
\begin{align}
	\label{eq:defineMIntro}
	M(h,X) = \max_{y \in [m]} \abs{\set{x \in X \mid h(x) = y}}
	\, .
\end{align}
In \Cref{thm:modP} we prove that $\Ep{M(h,X)} = \Oh{\sqrt[3]{n\log n}}$. We also consider
the hash function $\hb : [q] \to [m]$ defined by $\hb(x) = \floor{\frac{(ax) \bmod q}{q/m}}$,
where $q,m$ are powers of $2$, $q \ge m \ge n$ and $a$ is chosen uniformly at random among the
odd numbers from $[q]$. The function $\hb(x)$ was first introduced by
Dietzfelbinger et al. \cite{dietzfelbinger1997reliable}.
In \Cref{thm:multiplyShift} we prove that it also holds that
$\Ep{M(\hb,X)} = \Oh{\sqrt[3]{n \log n}}$.

We note that when we use $h(x) = ((ax+b) \bmod p) \bmod m$ in hashing with chaining, $M$
is the size of the largest chain. When scanning the hash table for an element the expected
time used is $\Oh{1}$ and the worst case time is at most $\Oh{M(h,X)}$. 

\subsection{Related work}
\newcommand{\F}{\mathcal{F}}

It is folklore that the size of the largest chain is $\Oh{\sqrt{n}}$ and this bounds hold
for any $2$-independent hash function.

Alon et al. \cite{isLinearHashingGood} considers the linear hash function $h_{m,k} : \F^m \to \F^k$,
where $\F$ is a finite field and $n = \abs{\F}^k$. The function is defined by
$h_{m,k}(x_1,\ldots,x_m) = \sum_i x_i a_i$, where $a_i \in \F^k$ is chosen uniformly
at random. For $m=2,k=1$ the hash function is $h_{2,1}(x,y) = ax+by$ where $a,b \in \F$
are chosen uniformly at random. It is shown in \cite{isLinearHashingGood} that there exists
a set $X \subset \F^2$ such that $\Ep{M(h_{2,1},X)} > \sqrt{n}$ if $n$ is a square and
$\Ep{M(h_{2,1},X)} = \Omegah{\sqrt[3]{n}}$ if $n$ is a prime power that is not a square.
In \cite{isLinearHashingGood} it is also shown that when $\F$ is the field of two elements the
expected length of the longest chain is $\Oh{\log n \log \log n}$ improving the results in
\cite{markowskyCW1978analysis,mehlhornVishkin1984randomized}.

Broder et al. \cite{broderCFM2000minwise} considered $h(x) = (ax+b) \bmod p$ in the context
of min-wise hashing.



\section{Preliminaries}


$\Z$ denotes the integers, and $\Z_n = \Z/n\Z$ denotes the integers
$\bmod\ n$. $\Z_n^*$ is the set of elements of $\Z_n$ having a multiplicative
inverse. $[n]$ is the set of integers from $0$ to $n-1$, that is
$[n] = \set{0,1,2,\ldots,n-1}$.
For a pair of integers $n,m \in \Z$ such that $(n,m) \neq (0,0)$ we let $\gcd(n,m)$
denote the greatest common divisor of $n$ and $m$. If $\gcd(n,m) = 1$ then $n$ and $m$
are said to be \emph{coprime}.

For integers $x,r \in \Z$ we let $\proj{r}{x} \in \Z_r$ denote the residue class of $x \bmod r$.
We let $\embDef{r} : \Z_r \to [r]$ be the unique mapping that satisfies $\proj{r}{\emb{r}{x}} = x$.
For $x \in \Z_r$ we let $\modNorm{r}{x} = \min\set{\emb{r}{x}, \emb{r}{-x}}$.

For a set $S$ and an element $x$ the sets $S+x$ and $xS$ are defined as $\set{s+x \mid s \in S}$ and
$\set{xs \mid s \in S}$, respectively.

For integers $r,s,m$, let $I_m(r,s) \subset \Z_m$ denote the set
\begin{align*}
	I_m(r,s) = \set{\proj{m}{r}, \proj{m}{r+1}, \ldots, \proj{m}{r+s-1}}
	\, .
\end{align*}
The set $I_m(r,s)$
is called an \emph{interval}. A non-empty set $X \subset \Z_m$ is an interval if there 
exists $r,s$ such that $X = I_m(r,s)$.
\section{Main Result}

In this section we prove the main results of this paper, namely \Cref{thm:modP}
and \Cref{thm:multiplyShift}. The proofs of the two theorems are very similar and 
both rely on \Cref{lem:technical} below.
\begin{lemma}
	\label{lem:technical}
	Let $n,M,r$ be integers satisfying $4M \le n \le r$ and let $A \subset \Z_r$ be a set
	of size $n$. Let $B \subset \Z_p^*$ be a set of size $\le M$ satisfying the
	following conditions:
	\begin{itemize}
		\item[1] $\iota(b) \in \left(M,2M\right)$ for all $b \in B$.
		\item[2] $\iota(b), \iota(b')$, and $r$ are pairwise coprime for every $b,b' \in B$ with $b \neq b'$.
	\end{itemize}
	Assume that for every $b \in B$ there exists an interval $I_b$ of size $\ceil{\frac{r}{n}}$
	such that $I_b \cap bA$ contains at least $4M$ elements. Then there exists at least
 	$M\abs{B}$ ordered pairs of different elements $a,a' \in A$ such that
	$\abs{a-a'} < \frac{r}{nM}$.
\end{lemma}
\begin{proof}
We note that for every $b$ the set $b^{-1}I_b$ is the union of $\iota(b)$ disjoint intervals of
size $\le \ceil{\frac{r}{n\iota(b)}}$, and we write it as such a union
$b^{-1}I_b = \bigcup_{j=0}^{\iota(b)-1} I_{b,j}$. For any $b,b' \in B, b \neq b'$	
the set $b^{-1}I_b \cap {b'}^{-1}I_{b'}$ is either empty or an interval.
So for each $b,b' \in B, b \neq b'$ there is at most one index $j \in [\iota(b)]$ such that the
intersection $I_{b,j} \cap {b'}^{-1}I_{b'}$ is non-empty. For every $b \in B$ and $j \in [\iota(b)]$,
let $\delta(b,j)$ denote the number of elements $b' \in B$ such that $I_{b,j} \cap {b'}^{-1}I_{b'}$ is non-empty.
Note that $\delta(b,j) \ge 1$ since $b \in B$. Furthermore $\sum_{j=0}^{\iota(b)-1} \delta(b,j) < \abs{B}+\iota(b) \le 3M$
since each ${b'}^{-1}I_{b'}$ has a non-empty intersection with at most one of the sets $I_{b,j}, j \in [\iota(b)]$.

The number of ordered pairs of different elements $(a,a') \in A \cap I_{b,j}$ such that $\abs{a-a'} < \frac{r}{nM}$
is exactly $\abs{A \cap I_{b,j}} \cdot \left ( \abs{A \cap I_{b,j}} - 1 \right )$ since $I_{b,j}$ is an
interval of size $\le \ceil{\frac{r}{n\iota(b)}}$ and $\iota(b) > M$.
Let $\tau(b,j) = \max\set{0,\abs{A \cap I_{b,j}}-1}$, then
the number of pairs is at least $(\tau(b,j))^2$.
We can lower bound the number of such pairs in $A$ by considering the pairs in $A \cap I_{b,j}$ for each
$b \in B$ and $j \in [\iota(b)]$ and note that each pair we count is counted at most $\delta(b,j)$ times.
This gives that the number of ordered pairs $(a,a') \in A$ such that $\abs{a-a'} < \frac{r}{nM}$
is at least:
\begin{align}
	\label{eq:sumLowerBound}
	\sum_{b \in B}
		\sum_{j \in [\iota(b)]}
			\frac{(\tau(b,j))^2}{\delta(b,j)}
\end{align}
For any $b \in B$, by the Cauchy-Schwartz inequality we have that:
\begin{align}
	\label{eq:cauchySchwartz}
	\left (
		\sum_{j \in [\iota(b)]}
			\delta(b,j)
	\right )
	\left (
		\sum_{j \in [\iota(b)]}
			\frac{(\tau(b,j))^2}{\delta(b,j)}
	\right )
	\ge 
	\left (
		\sum_{j \in [\iota(b)]}
			\tau(b,j)
	\right )^2
\end{align}
We clearly have that $\sum_{j \in [\iota(b)]} \tau(b,j) \ge 4M-\iota(b) \ge 2M$. Also recall, that
we have that
$\sum_{j \in [\iota(b)]} \delta(b,j) \le 3M$. Combining this with \eqref{eq:sumLowerBound}
and \eqref{eq:cauchySchwartz} gives 
that $A$ contains at least $\frac{4M\abs{B}}{3} \ge M\abs{B}$ of the desired pairs.
\end{proof}

Below is a proof of \Cref{thm:modP}.
\begin{theorem}
	\label{thm:modP}
	Let $n,m,p$ be integers with $p$ a prime and $p \ge m \ge n$. Let $X \subset \Z_p$
	be a set of $n$ elements. Let $h : \Z_p \to \Z_m$ be defined by
	$h(x) = \proj{m}{\emb{p}{ax+b}}$
	where $a,b \in \Z_p$ are chosen uniformly at random.
	Let $M = M(X)$ be the random variable counting the number of elements $x \in X$ that
	hash to the most popular hash value, that is
	\begin{align*}
		M = M(X) = \max_{y \in \Z_m} \abs{\set{x \in X \mid h(x) = y}}
		\, .
	\end{align*}
	Then
	\begin{align}
		\label{eq:exOfMModP}
		\Ep{M} = O\!\left ( \sqrt[3]{n \log n} \right )
		\, .
	\end{align}
\end{theorem}
\begin{proof}
We note that $\Ep{M \mid a=0} = n$ since $h$ is constant when $a=0$. Therefore:
\begin{align}
	\notag
	\Ep{M} = 
	\frac{p-1}{p} \Ep{M \mid a \neq 0} + 
	\frac{1}{p} \Ep{M \mid a = 0}
	< 
	\Ep{M \mid a \neq 0} + 1
	\, .
\end{align}
Therefore it suffices to bound the expected value of $M$ when $a$ is chosen uniformly
at random from $\Z_p \setminus \set{0} = \Z_p^*$ and not from $\Z_p$. So from now on, assume
that $a$ is chosen uniformly at random from $\Z_p^*$.

The random variables $a$ and $a^{-1}b$ are independent. Note, that $h(x)$ can be rewritten as
$h(x) = \proj{m}{\emb{p}{a(x+a^{-1}b)}}$. It clearly suffices to bound the expected value of $M$
conditioned on all possible values $a^{-1}b$. For any fixed value of $a^{-1}b = c$, the expected value
of $M$ conditioned on $a^{-1}b = c$ is the same as the expected value of $M(X+c)$ conditioned
on $b = 0$. Therefore it suffices to give the proof under the assumption that $b = 0$. So
we assume that $b = 0$.

Let $A = m^{-1}aX$, then there exists an interval $I_a$ of size at most $\ceil{\frac{p}{m}}$
that contains $M$ elements of $A$ for the following reason: Let $f : \Z_p \to \Z_m$ be defined
by $x \to \proj{m}{\emb{p}{x}}$. By definition, there exists a random variable $y \in \Z_p$ such
that $\abs{f^{-1}(y) \cap aX} \ge M$. And there exists a $i \in [m]$ such that
\begin{align*}
	f^{-1}(y) = \set{\proj{p}{i+km} \mid k \in \Z, 0 \le k < \frac{p-i}{m}}
	\, ,
\end{align*}
and hence $I_a = m^{-1}f^{-1}(y)$ is an interval of size $\le \ceil{\frac{p}{m}}$ that contains
$M$ elements of $A$.

Let $\alpha \in \left[1,\frac{n}{4}\right]$. We are now going to bound the probability that $M \ge 4\alpha$. Let 
$\delta = \Prp{M \ge 4\alpha}$ and let $\A$ be the set of all
elements $a_0 \in \Z_p^*$ such that $M \ge 4\alpha$ if $a = a_0$.

Let $S \subset \Z_p^*$ be the set of all elements $s \in \Z_p^*$
that satisfies that $\emb{p}{s}$ is a prime in the interval $(\alpha,2\alpha)$.
Let $B \subset S$ be the set of all elements $s \in S$ such that $as \in \A$. Note, that $B$
is a random variable. By linearity of expectation, we have that $\Ep{\abs{B}} = \abs{S}\delta$.
Recall, that $A = m^{-1}aX$. For any $b \in B$ we have that $ab \in \A$ and therefore there exists
an interval of size $\ceil{\frac{r}{n}}$ that contains at least $4\alpha$ elements of $bA$.
By \Cref{lem:technical}, this implies that there are $\alpha\abs{B}$ ordered pairs of
different elements $x,x' \in X$ such that $\modNorm{p}{ax-ax'} < \frac{p}{m\alpha}$.
So the expected number of elements $x,x' \in X$ such that $\modNorm{p}{a(x-x')} < \frac{p}{m\alpha}$
is at least $\alpha\Ep{\abs{B}} = \alpha \delta \abs{S}$. On the other hand, for each
ordered pair of different elements $x,x' \in X$ the probability that $\modNorm{p}{a(x-x')} < \frac{p}{m\alpha}$
is at most $\frac{2p}{m\alpha(p-1)}$, and by linearity of expectation the expected number of such
ordered pairs is at most
\begin{align*}
	n(n-1) \cdot \frac{2p}{m\alpha(p-1)}
	\le 
	\frac{2n}{\alpha}
	\, .	
\end{align*}
We conclude that $\alpha \delta \abs{S} \le \frac{2n}{\alpha}$. By the prime
number theorem, $\abs{S} = \Theta\!\left(\frac{\alpha}{\log \alpha}\right) =
\Omega\!\left(\frac{\alpha}{\log n}\right )$.
Reordering gives us that:
\begin{align}
	\notag
	\Prp{M \ge 4\alpha} =
	\delta =
	O\!\left (
		\frac{n \log n}{\alpha^3}
	\right )
	\, .
\end{align}
The expected value of $M$ can now be bounded in the following manner:
\begin{align*}
	\Ep{M}
	& =
	\sum_{k = 1}^{\infty} \Prp{M \ge k}
	\\
	& =
	\sum_{k = 1}^{\floor{\sqrt[3]{n \log n}}} \Prp{M \ge k}
	+
	\sum_{k = \floor{\sqrt[3]{n \log n}}+1}^{n} \Prp{M \ge k}
	\\
	& \le
	\floor{\sqrt[3]{n \log n}}
	+
	\sum_{k = \floor{\sqrt[3]{n \log n}}+1}^{n} 
	O\!\left (
		\frac{n\log n}{k^3}
	\right )
	\\
	& =
	O\!\left (
		\sqrt[3]{n \log n}
	\right )
\end{align*}
which was what we wanted.
\end{proof}

The proof of \Cref{thm:multiplyShift} is very similar to the proof
of \Cref{thm:modP} but we include it for completeness.
\begin{theorem}
	\label{thm:multiplyShift}
	Let $n,\ell,r,q,m$ be integers with $q = 2^r, m = 2^\ell$ and $q \ge m \ge n$. Let $X \subset \Z_q$
	be a set of $n$ elements. Let $h : \Z_q \to [m]$ be defined by
	$h(x) = \floor{\emb{q}{ax} \cdot 2^{\ell-r}}$ 
	where $a \in \Z_q^*$ are chosen uniformly at random.
	Let $M = M(X) = \max_{y \in [m]}\abs{\set{x \in X \mid h(x) = y}}$. Then
	\begin{align}
		\label{eq:exOfMMulShift}
		\Ep{M} = O\!\left ( \sqrt[3]{n \log n} \right )
		\, .
	\end{align}
\end{theorem}
\begin{proof}
Let $y$ be a random variable such that $\abs{h^{-1}(y) \cap X} = M$, and let $A = aX$.
The set $ah^{-1}(y)$ is an interval of size $\frac{q}{m}$ that contains exactly $M$
elements of $A$. 

Let $\alpha \in \left[1,\frac{n}{4}\right]$. We are now going to bound the probability that $M \ge 4\alpha$. Let 
$\delta = \Prp{M \ge 4\alpha}$, and let $\A$ be the set of all elements
$a_0 \in \Z_p^*$ such that $M \ge 4\alpha$ if $a = a_0$.

Let $S \subset \Z_q^*$ be the set of all elements $s \in \Z_q^*$
that satisfies that $\emb{q}{s}$ is a prime in the interval $(\alpha,2\alpha)$.
Let $B \subset S$ be the set of all elements $s \in S$ such that $as \in \A$. Note, that $B$
is a random variable. By linearity of expectation, we have that $\Ep{\abs{B}} = \abs{S}\delta$.
Recall, that $A = m^{-1}aX$. For any $b \in B$ we have that $ab \in \A$ and therefore there exists
an interval of size $\frac{q}{m}$ that contains at least $4\alpha$ elements of $bA$.
By \Cref{lem:technical}, this implies that there are $\alpha\abs{B}$ ordered pairs of
different elements $x,x' \in X$ such that $\modNorm{q}{ax-ax'} < \frac{q}{m\alpha}$.
So the expected number of elements $x,x' \in X$ such that $\modNorm{q}{a(x-x')} < \frac{q}{m\alpha}$
is at least $\alpha\Ep{\abs{B}} = \alpha \delta \abs{S}$. On the other hand, for each
ordered pair of different elements $x,x' \in X$ the probability that $\modNorm{q}{a(x-x')} < \frac{q}{m\alpha}$
is at most $\frac{4}{m\alpha}$, and by linearity of expectation the expected number of such
ordered pairs is at most
\begin{align*}
	n(n-1) \cdot \frac{4}{m\alpha}
	\le 
	\frac{4n}{\alpha}
	\, .
\end{align*}
We conclude that $\alpha \delta \abs{S} \le \frac{4n}{\alpha}$, and now we can bound the
expected value exactly as in \Cref{thm:modP}.
\end{proof}

\newpage

\bibliographystyle{plain}
\bibliography{bib}

%
%


\end{document}